\documentclass[10pt,journal,twoside,letterpaper]{IEEEtran}

\usepackage{amsmath,amssymb,amsthm,amsfonts}
\usepackage{mathtools}
\usepackage{graphicx}
\usepackage[caption=false, font=footnotesize]{subfig}
\usepackage{booktabs}
\usepackage{tikz}
\usetikzlibrary{
   positioning, arrows.meta, calc, fit,
   shapes.geometric, shapes.misc, decorations.pathmorphing
}

\usepackage{cite}
\usepackage{url}
\usepackage{hyperref,enumitem}
\hypersetup{hidelinks}

\newtheorem{theorem}{Theorem}
\newtheorem{corollary}{Corollary}
\newtheorem{definition}{Definition}

\newtheorem{remark}{Remark}

\newcommand{\Pss}{\bar{\mathbf{P}}^s}
\newcommand{\hatxsk}{\hat{\mathbf{x}}^s_{k|k}}
\newcommand{\zhat}[1]{z_{#1}^s}
\newcommand{\confusion}{\mathcal{C}}
\newcommand{\zminus}{z^-}
\newcommand{\zplus}{z^+}
\newcommand{\Te}{T_e}
\newcommand{\PoutZero}{p_{0,\mathrm{out}}}

\begin{document}

\title{React or Predict? A Spectral Rule for Wireless Threshold Detection}

\author{
Aamir~Mahmood,~\IEEEmembership{Senior Member,~IEEE,}
Nho Duc~Tran,~\IEEEmembership{Student Member,~IEEE,}     
\thanks{A. Mahmood and T.N. Duc are with the Department of Computer and Electrical Engineering,  Mid Sweden University, Homlgatan 10, 851 70 Sundsvall, Sweden (e-mail: \{aamir.mahmood, nhoduc.tran\}@miun.se).}
}

\maketitle


\begin{abstract} 
A wireless sensor must alert a remote monitor before a monitored process crosses a safety threshold; an alarm arriving afterward may be too late. The sensor can react to its current estimate or predict ahead and trigger earlier, but the value of such lookahead is not obvious. In some systems it creates an early-alarm opportunity unavailable to the current test, while in others it cannot cross the alarm boundary. This letter gives a practical three-stage rule for deciding when to predict. First, an algebraic spectral test decides at design time whether lookahead is structurally useful: it is redundant exactly when the threshold direction is a left-eigenvector of the dynamics with a non-negative eigenvalue. Second, a closed-form channel decomposition shows that deeper prediction becomes more valuable as the channel degrades, because longer lead windows permit more pre-crossing transmission attempts. Third, simulations show that large gains also require retained prediction magnitude; oscillatory dynamics amplify the benefit through rotation, and a two-sensor setting reveals a sensing-channel tradeoff. \end{abstract}



\section{Introduction}\label{sec:intro}

\IEEEPARstart{M}{any} safety-critical remote monitoring systems share a common control loop: a wireless sensor observes a physical process and must alert a remote monitor when a threshold is crossed. Such loops are pervasive in Industrial IoT, from vibration monitoring and structural health to medical telemetry. In all of them, the difficulty is not only detecting the crossing, but delivering the alarm before it, early enough for the operator to react. Because the wireless link is unreliable, a one-shot transmission may not survive, which makes the alarm transmission policy itself part of the design problem (Fig.~\ref{fig:system_model}).

A natural way to buy time is to \emph{predict ahead}, transmitting early so the alarm has several attempts to survive packet loss before the crossing (Fig.~\ref{fig:system_model}, lower panel). Recent work~\cite{tran_pred} introduced a predictive triggering framework that performs $H$-step prediction at the sensor, combined with age of information (AoI)-controlled refresh and a two-state Markov surrogate for tractable analysis. This sits within a broader literature on event-triggered estimation~\cite{heemels2012}, AoI~\cite{kaul2012,yates2021age}, age of incorrect information (AoII)~\cite{maatouk2020aoii}, goal-oriented communication~\cite{fountoulakis2023}, and remote estimation over unreliable channels~\cite{quevedo2013,nayyar2018}, which primarily optimizes when to transmit and how to trade estimation error against communication cost. However, the prior work has not explicitly characterized when lookahead itself creates threshold-decision information that current-only cannot access.

This raises a design question that is invisible at run time. Looking further ahead increases prediction variance, yet two systems with identical dynamics and wireless links can respond very differently to a longer horizon, depending on the monitored direction $\mathbf{c}$. For some $\mathbf{c}$, the horizon is irrelevant while for others it is a genuine design knob. The distinction is geometric rather than obvious from the channel parameters alone. The sensor must therefore choose whether to react to its current estimate or predict ahead and trigger early.
\begin{figure}
    \centering
    \includegraphics[width=0.85\linewidth]{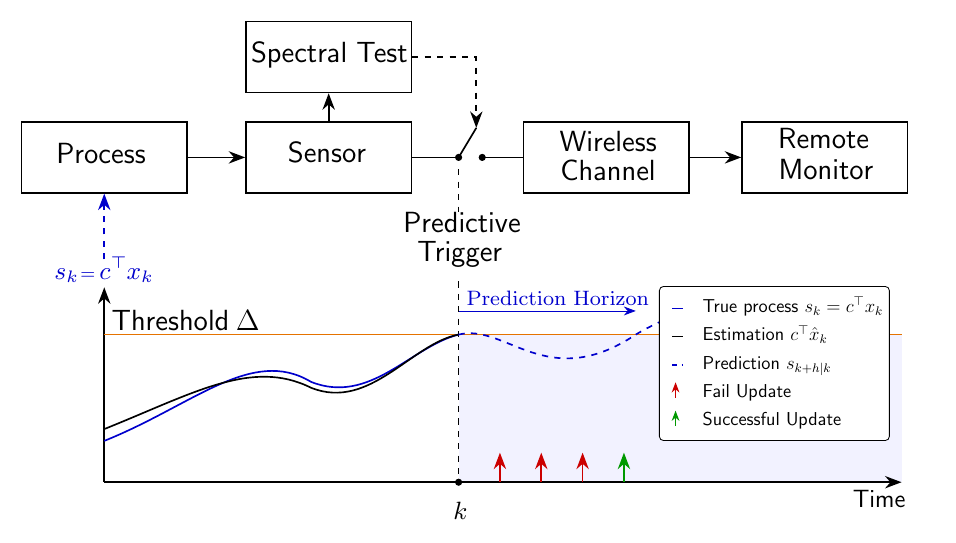}
    \caption{System model and contribution: when the spectral test finds lookahead useful, the sensor shifts from reacting to preemptive triggering, forecasting the crossing $H$ steps ahead to buy transmission attempts before the event.}
    \label{fig:system_model}
    \vspace{-10pt}
\end{figure}
This letter asks when prediction can create an early-alarm opportunity unavailable to the current test (Fig.~\ref{fig:system_model}, upper panel). The answer is spectral: lookahead is redundant exactly when $\mathbf{c}^\top\mathbf{A}=\lambda\mathbf{c}^\top$ with $\lambda\geq0$; the test checks collinearity and, when aligned, the sign of $\lambda$. We then derive a closed-form decomposition showing that, once lookahead is meaningful, deeper prediction becomes more valuable as the channel degrades. Simulations sharpen the picture: locked systems show no benefit, unlocked systems can, but a large gain also requires that the dynamics retain prediction magnitude over the horizon, with oscillatory modes amplifying it through rotation. Finally, a two-sensor extension trades sensing noise for channel diversity and widens the benefit of deeper lookahead.

\section{Problem Setup}\label{sec:model}

\subsection{Early Detection and the Performance Metric}
\label{sec:process}
We monitor a physical process whose state evolves as a stable linear system,
$\mathbf{x}_{k+1} = \mathbf{A}\mathbf{x}_k + \mathbf{w}_k$, with
$\mathbf{w}_k \sim \mathcal{N}(\mathbf{0},\mathbf{Q}_w)$,
$\mathbf{x}_k \in \mathbb{R}^n$, and $\rho(\mathbf{A}) < 1$.
The safety-relevant quantity is the scalar projection
$s_k = \mathbf{c}^\top \mathbf{x}_k$ along a monitored direction
$\mathbf{c} \in \mathbb{R}^n$. A threshold crossing occurs at time $\Te$ when
$s_{\Te-1} < \Delta$ and $s_{\Te} \geq \Delta$. The pair $(\mathbf{A},\mathbf{c})$
governs the geometry of future crossings, a fact the spectral
characterization of Sec.~\ref{sec:main} turns into a design test.

The object of interest is not the crossing itself but its timing. Let $T_{\mathrm{alarm}}$ be the slot at which the remote monitor raises the alarm, and define the lead time $L := \Te - T_{\mathrm{alarm}}$: a positive lead means the alarm arrives before the crossing. The performance metric is the probability of early detection, $P(L>0)$, taken over the process and measurement noise and the random channel outcomes. Achieving $L > 0$ requires two things in sequence: the sensor must decide to alarm early, and that alarm must survive the wireless link, the structural and operational questions that Sec.~\ref{sec:main} answers in turn.

\subsection{Sensor-Side Prediction and the Predictive Trigger}
The sensor measures $y_k^s = \mathbf{c}^\top \mathbf{x}_k + v_k^s$ with
$v_k^s \sim \mathcal{N}(0,R_v^s)$ and maintains a Kalman filter with posterior
$(\hatxsk,\mathbf{P}^s_{k|k})$; detectability of $(\mathbf{A},\mathbf{c}^\top)$
guarantees a steady-state posterior covariance $\Pss$. To decide early, the
sensor predicts the projection forward. At horizon $j \in \{0,1,\ldots,H\}$, it
forms the normalized confidence
\begin{equation}\label{eq:zj}
    \zhat{j}(k) := \bigl(\Delta - \mathbf{c}^\top \mathbf{A}^j \hatxsk\bigr)
    \big/ \sqrt{\mathbf{c}^\top \mathbf{P}^s_{k+j|k} \mathbf{c}},
\end{equation}
with $\mathbf{P}^s_{k+j|k} = \mathbf{A}^j \Pss (\mathbf{A}^j)^\top
+ \sum_{i=0}^{j-1} \mathbf{A}^i \mathbf{Q}_w (\mathbf{A}^i)^\top$.

The decision is read through a confusion region~\cite{tran_pred}
$\confusion := (\zminus,\zplus)$, with $\zminus = \Phi^{-1}(\alpha_{\mathrm{FP}})$
and $\zplus = \Phi^{-1}(1-\alpha_{\mathrm{FN}})$ for false-positive and
false-negative tolerances $(\alpha_{\mathrm{FP}},\alpha_{\mathrm{FN}})\in(0,0.5)^2$.
A score $\zhat{j} \le \zminus$ supports the alarm decision, $\zhat{j} \ge \zplus$
supports no-alarm, and a score inside $\confusion$ is inconclusive. The trigger is
task-oriented: the sensor transmits at slot $k$ whenever some current or predicted
score supports alarm, i.e., $\zhat{j}(k)\le\zminus$ for some
$j\in\{0,\ldots,H\}$. Lookahead creates an additional early-alarm opportunity when
$\zhat{0}(k)>\zminus$ but $\zhat{j}(k)\le\zminus$ for some $j\in\{1,\ldots,H\}$. The sensor retransmits the alarm each slot until it is delivered (assuming idealized delivery feedback), so a lead of $j$ gives $j$ attempts and success probability $1-p_{\mathrm{fail}}^j$; imperfect feedback would only reduce this.

\subsection{Wireless Channel}
A transmission attempt may fail via two mechanisms on different time scales. On the slow scale, the link may enter an \emph{outage}, a shadowing or blockage interval during which it is unusable. A blockage begins on a given slot with probability $\PoutZero$ and persists for a Weibull-distributed duration, so consecutive attempts within a blockage are lost together. On the fast scale, when the link is available, the packet undergoes small-scale Rayleigh fading with power gain $g$ and fails decoding with the finite-blocklength error probability $p_{\mathrm{PER}}(\gamma)$~\cite{polyanskiy2010} at the instantaneous SNR $\gamma$. Averaging over the fading, the marginal one-attempt failure probability is
\begin{equation}\label{eq:pfail}
    p_{\mathrm{fail}}
    = \pi_{\mathrm{out}}(\PoutZero)
    + \bigl(1-\pi_{\mathrm{out}}(\PoutZero)\bigr)\,
    \mathbb{E}_g\!\left[p_{\mathrm{PER}}(\gamma)\right],
\end{equation}
combining blockage and decoding error in the standard two-cause form for short-packet links. Here $\pi_{\mathrm{out}}(\PoutZero)$ is the marginal outage probability induced by the onset probability $\PoutZero$ and the Weibull burst-duration law. We use $\PoutZero$ as the simulation channel-stress parameter; increasing $\PoutZero$ monotonically increases $\pi_{\mathrm{out}}$. The two mechanisms correlate differently across attempts. Fading is renewed each attempt, so each additional pre-crossing attempt is a fresh chance to succeed; blockage is temporally correlated, and attempts within one outage fail together. The analysis of Sec.~\ref{sec:main} nonetheless treats attempts as independent with marginal $p_{\mathrm{fail}}$; Sec.~\ref{sec:empirical} quantifies the gap.


\section{Spectral Characterization}\label{sec:main}

The usefulness of lookahead depends on how the dynamics transform the monitored direction. Non-collinearity creates a new threshold direction, while aligned dynamics yield a one-dimensional affine map whose sign determines whether prediction can cross the alarm boundary.

\subsection{When does lookahead create a new opportunity?}

Suppose the monitored direction is aligned with the dynamics in the sense that
$\mathbf{c}^\top \mathbf{A} = \lambda \mathbf{c}^\top$
for some scalar $\lambda$. Then, by induction,
\begin{equation}
    \mathbf{c}^\top \mathbf{A}^j = \lambda^j \mathbf{c}^\top, \; j\geq 0.
\end{equation}
Substituting into~\eqref{eq:zj} shows that every predicted $z$-score is an affine function of the current one. Thus aligned prediction remains one-dimensional; whether it creates a new alarm depends on the sign of $\lambda$ and the resulting affine map. If $\mathbf{c}^\top\mathbf{A}$ is not collinear with $\mathbf{c}^\top$, the current and one-step predicted tests are different linear functionals of the state, allowing a current non-alarm to become a predicted alarm.

\begin{definition}[Improvement set]
The improvement set $\mathcal{I}_H$ contains posterior states for which the current test does not support alarm but some predicted test does:
$\mathcal{I}_H := \{\hat{\mathbf{x}}:\zhat{0}(\hat{\mathbf{x}})>\zminus,\ \exists j\in\{1,\ldots,H\}\text{ such that }\zhat{j}(\hat{\mathbf{x}})\le\zminus\}$.
\end{definition}

If $\mathcal{I}_H\neq\varnothing$, lookahead can create an early-alarm opportunity unavailable to the current-only trigger; if $\mathcal{I}_H=\varnothing$, it cannot.

\begin{theorem}[Spectral characterization]\label{thm:spectral}
Consider the predictive triggering setting of Sec.~\ref{sec:model} with stable $\mathbf{A}$, detectable $(\mathbf{A},\mathbf{c}^\top)$, steady-state posterior $\bar{\mathbf{P}}^s$, and the operationally meaningful detection condition $\Delta/\sigma_\infty \in \mathcal{C}$, where $\sigma_\infty := \sqrt{\mathbf{c}^\top \mathbf{P}^x_\infty \mathbf{c}}$
and $\mathbf{P}^x_\infty = \sum_{i=0}^{\infty}\mathbf{A}^i\mathbf{Q}_w(\mathbf{A}^i)^\top$
is the open-loop steady-state covariance.
Then:
\begin{enumerate}[label=(\alph*),leftmargin=*]
    \item
          If $\mathbf{c}^\top \mathbf{A} \neq \lambda \mathbf{c}^\top$
          for every $\lambda \in \mathbb{R}$,
          then $\mathcal{I}_H \neq \varnothing$ for every $H \geq 1$.
    \item 
          If $\mathbf{c}^\top \mathbf{A} = \lambda \mathbf{c}^\top$
          for 
          $\lambda \geq 0$,
          then $\mathcal{I}_H = \varnothing$ for every $H \geq 1$.
    \item
          If $\mathbf{c}^\top \mathbf{A} = \lambda \mathbf{c}^\top$
          for $\lambda < 0$,
          then $\mathcal{I}_H \neq \varnothing$ for every $H \geq 1$.
\end{enumerate}
\end{theorem}

\begin{proof}
Denote $\sigma_0 := \sqrt{\mathbf{c}^\top \Pss \mathbf{c}}$
and $\sigma_j := \sqrt{\mathbf{c}^\top \mathbf{P}^s_{k+j|k} \mathbf{c}}$.

\emph{(a).}
If $\mathbf{c}^\top\mathbf{A} \neq \lambda\mathbf{c}^\top$ for any $\lambda$,
then $\mathbf{c}^\top$ and $\mathbf{c}^\top\mathbf{A}$ are linearly independent
in $(\mathbb{R}^n)^*$.
The map $T\colon\mathbb{R}^n\to\mathbb{R}^2$,
$T(\hat{\mathbf{x}}) = \bigl(\mathbf{c}^\top\hat{\mathbf{x}},\,
\mathbf{c}^\top\mathbf{A}\hat{\mathbf{x}}\bigr)$,
has rank~$2$ and is surjective.
Choose $u^* = \Delta$ and $v^* = \Delta + 2|z^-|\sigma_1$.
Surjectivity yields $\hat{\mathbf{x}}^*$ with $T(\hat{\mathbf{x}}^*)=(u^*,v^*)$,
giving $\hat{z}_0(\hat{\mathbf{x}}^*) = 0 > z^-$
and $\hat{z}_1(\hat{\mathbf{x}}^*) = -2|z^-| < z^-$,
so $\hat{\mathbf{x}}^* \in \mathcal{I}_1 \subseteq \mathcal{I}_H$.

\emph{(b).}
Suppose $\mathbf{c}^\top\mathbf{A} = \lambda\mathbf{c}^\top$ with $\lambda \geq 0$.
By induction $\mathbf{c}^\top\mathbf{A}^j = \lambda^j\mathbf{c}^\top$,
so substituting into~\eqref{eq:zj} gives
\begin{equation}\label{eq:affine}
    \hat{z}_j(\hat{\mathbf{x}})
      = a_j\,\hat{z}_0(\hat{\mathbf{x}}) + (1 - a_j)\,p_j, \;
    a_j := \lambda^j\sigma_0/\sigma_j \in [0,1),
\end{equation}
where $p_j := \Delta(1-\lambda^j)/(\sigma_j - \lambda^j\sigma_0)$
for $j \geq 1$.
Because $a_j \in [0,1)$, equation~\eqref{eq:affine} is a convex combination
of $\hat{z}_0$ and $p_j$; it therefore suffices to show $p_j \in \mathcal{C}$.
The aligned structure $\mathbf{c}^\top\mathbf{A}^i = \lambda^i\mathbf{c}^\top$
yields the closed-form variance identity\footnote{From $(\mathbf{A}^i)^\top\mathbf{c} = \lambda^i\mathbf{c}$, the recursion gives $\sigma_j^2 = \lambda^{2j}\sigma_0^2 + (\mathbf{c}^\top\mathbf{Q}_w\mathbf{c})\sum_{i=0}^{j-1}\lambda^{2i}$; letting $j\to\infty$ gives $\sigma_\infty^2 = (\mathbf{c}^\top\mathbf{Q}_w\mathbf{c})/(1-\lambda^2)$, and substituting back yields~\eqref{eq:sigma_j_aligned}.}
\begin{equation}\label{eq:sigma_j_aligned}
    \sigma_j^2 \;=\; \sigma_\infty^2 - \lambda^{2j}\!\left(\sigma_\infty^2 - \sigma_0^2\right),
\end{equation}
where $\sigma_0 \leq \sigma_\infty$, so $\sigma_j$ is non-decreasing on $[\sigma_0,\sigma_\infty]$.
Since $\lambda \geq 0$, we have $\operatorname{sign}(p_j) = \operatorname{sign}(\Delta)$,
so $p_j$ lies on the same side of~$0$ as $\Delta/\sigma_\infty$.
It remains to establish $|p_j| \leq |\Delta/\sigma_\infty|$,
equivalently $\sigma_\infty - \sigma_j \leq \lambda^j(\sigma_\infty - \sigma_0)$. If $\sigma_\infty=\sigma_0$, then~\eqref{eq:sigma_j_aligned} gives
$\sigma_j=\sigma_\infty$ for all $j$, so
$p_j=\Delta/\sigma_\infty\in\mathcal{C}$ and the claim follows. Otherwise, 
from~\eqref{eq:sigma_j_aligned},
\[
    \sigma_\infty - \sigma_j
    = \frac{\lambda^{2j}(\sigma_\infty^2-\sigma_0^2)}{\sigma_\infty+\sigma_j}
    = \lambda^{2j}(\sigma_\infty-\sigma_0)\,
      \frac{\sigma_\infty+\sigma_0}{\sigma_\infty+\sigma_j}.
\]
Substitute to $\sigma_\infty - \sigma_j \leq \lambda^j(\sigma_\infty - \sigma_0)$ and cancel $(\sigma_\infty - \sigma_0) > 0$,
the inequality reduces to
$\lambda^j(\sigma_\infty+\sigma_0)/(\sigma_\infty+\sigma_j) \leq 1$,
which holds because $\lambda^j \leq 1$ and $\sigma_j \geq \sigma_0$.
Hence $p_j$ lies strictly between~$0$ and $\Delta/\sigma_\infty$
(with $p_j = 0$ when $\Delta = 0$).
Since both $0 \in \mathcal{C}$ and $\Delta/\sigma_\infty \in \mathcal{C}$,
convexity of $\mathcal{C}$ gives $p_j \in \mathcal{C}$.
Since $p_j\in\mathcal{C}$, we have $p_j>z^-$. Hence, whenever
$\hat{z}_0>z^-$, the convex combination~\eqref{eq:affine} gives
$\hat{z}_j>z^-$ for every $j\geq1$, so
$\mathcal{I}_H=\varnothing$ for all $H\geq1$.

\emph{(c).}
If $\lambda<0$, then~\eqref{eq:affine} has $a_1<0$.
As $\hat{z}_0$ ranges over $\mathbb{R}$,
$\hat{z}_1=a_1\hat{z}_0+(1-a_1)p_1\to-\infty$
as $\hat{z}_0\to\infty$. Hence some posterior state satisfies
$\hat{z}_0>z^-$ and $\hat{z}_1\le z^-$,
so $\mathcal{I}_1\neq\varnothing$ and therefore
$\mathcal{I}_H\neq\varnothing$ for every $H\geq1$.
\end{proof}

The theorem shows that lookahead is structurally redundant only in the non-negative aligned case: a current non-alarm then remains non-alarm at every horizon. Otherwise, there exist posterior states for which prediction creates an alarm unavailable to the current test. This is an existence guarantee, not a guarantee of appreciable occurrence probability.

\begin{remark}\label{rem:cond_i}
If $\Delta/\sigma_\infty \notin \mathcal{C}$, the predicted $z$-score
$\hat{z}_j(\hat{\mathbf{x}}) \to \Delta/\sigma_\infty$ eventually exits~$\mathcal{C}$
for every posterior $\hat{\mathbf{x}}$, forcing a permanent alarm or non-alarm decision.
This contradicts the recurring threshold-crossing model of Sec.~\ref{sec:process},
so $\Delta/\sigma_\infty \in \mathcal{C}$ must hold in any operationally meaningful setting.
\end{remark}

\begin{corollary}[Scalar systems]\label{cor:scalar}
For $n = 1$, $c^\top A =  A\,c^\top$, so $c$
is automatically a left-eigenvector with eigenvalue $\lambda = A$.
\begin{itemize}[leftmargin=*]
    \item If $A \in [0,1)$: Theorem~\ref{thm:spectral}(b) applies and
          $\mathcal{I}_H=\varnothing$ for every $H\geq1$; predictive
          lookahead is structurally redundant.
    \item If $A \in (-1,0)$: Theorem~\ref{thm:spectral}(c) applies and
          $\mathcal{I}_H\neq\varnothing$ for every $H\geq1$.
\end{itemize}
\end{corollary}

\begin{corollary}[Oscillatory two-dimensional systems]\label{cor:complex}
For a real two-dimensional system whose eigenvalues form a non-real
complex-conjugate pair, no nonzero real $\mathbf{c}$ is a left-eigenvector of
$\mathbf{A}$. Hence predictive lookahead is structurally meaningful for every
real threshold direction.
\end{corollary}

\begin{proof}
A real left-eigenvector corresponds to a real eigenvalue.
If the spectrum is a complex-conjugate pair, no real left-eigenvector exists.
\end{proof}
This corollary is practically important as many IIoT signals (e.g., vibration, structural health) are naturally oscillatory. The same conclusion extends to higher-dimensional systems with spectrum containing only non-real eigenvalues; for mixed spectra (real and complex eigenvalues coexisting), the test is applied to $\mathbf{c}^\top\mathbf{A}$. Numerically, set $\lambda^\star=(\mathbf{c}^\top\mathbf{A}\mathbf{c})/(\mathbf{c}^\top\mathbf{c})$ and test $\|\mathbf{c}^\top\mathbf{A}-\lambda^\star\mathbf{c}^\top\|$ against a chosen tolerance; if aligned, the sign of $\lambda^\star$ completes the test.

\subsection{How does channel quality change the value of lookahead?}

The structural result above tells us whether lookahead can help at all. It does not tell us how much of that structural opportunity survives a lossy wireless link. The operational part of the paper answers that question.

For a crossing at $T_e$, let $T_{\mathrm{trig}}$ be the first slot at which the trigger of Sec.~\ref{sec:model}-B fires, and let $q_H(j) := \Pr\{T_e - T_{\mathrm{trig}} = j\}$, $j \in \{1,\dots,H\}$, be the trigger-lead distribution; its mass need not sum to one, since with the remaining probability the trigger never fires before the crossing. This lead is set by the Kalman posterior and the threshold geometry alone, separate from whether the resulting transmissions are delivered. Under per-attempt independence, the probability that at least one of the $j$ transmission attempts succeeds is $1-p_{\mathrm{fail}}^j$. Therefore:

\begin{corollary}[Channel-conditioned operational value]\label{cor:channel}
Under the independent-attempt idealization with per-attempt failure probability $p_{\mathrm{fail}} \in (0,1)$, the probability of early detection decomposes as
\begin{equation}\label{eq:operational}
P(L>0 \mid H,p_{\mathrm{fail}})
= \sum_{j=1}^{H} q_H(j)\bigl(1-p_{\mathrm{fail}}^{\,j}\bigr),
\end{equation}
where $1-p_{\mathrm{fail}}^{\,j} = (1-p_{\mathrm{fail}})(1+p_{\mathrm{fail}}+\cdots+p_{\mathrm{fail}}^{\,j-1})$: relative to lead $1$, lead $j$ carries the weight $1+p_{\mathrm{fail}}+\cdots+p_{\mathrm{fail}}^{\,j-1}$, non-decreasing in $p_{\mathrm{fail}}$ with limit $j$. Thus, whenever $q_H$ places mass on leads $j\ge2$, those deeper-lead
contributions become relatively more valuable as the channel degrades.
\end{corollary}

The deeper lookahead acts as a rescue mechanism whose per-unit value scales with $p_{\mathrm{fail}}$. 
Eq.~\eqref{eq:operational} is a baseline independent-attempt expression, capturing the value of passive repetition over the lead window. Channel-aware use of the lead time, such as power control, rate adaptation, or scheduling against predicted channel states, may improve on it, whereas burst-correlated outage can reduce the independence gain, as quantified in Sec.~\ref{sec:empirical}.

\subsection{Scope and Limits of the Spectral Rule}
The structural theorem answers the yes-or-no question: Can prediction create an early-alarm opportunity unavailable to the current-only trigger? The channel corollary answers the operational question: if prediction is useful, how does the benefit depend on link quality? The remaining issue is quantitative: even when prediction is structurally allowed, is the benefit always large? The simulations below show that it is not always large, which motivates a final refinement based on how much prediction information survives over the horizon.

%
\begin{table}[t]
\centering
\caption{Simulation setup.}
\label{tab:setup}
\setlength{\tabcolsep}{4pt}
\footnotesize
\resizebox{0.85\columnwidth}{!}{
\begin{tabular}{@{}p{0.32\columnwidth} p{0.68\columnwidth}@{}}
\toprule
\multicolumn{2}{@{}l}{\textbf{System dynamics}} \\
Scalar         & $A = 0.9$; $c = 1$; $Q_w = 1$; $\Delta = 0.3\sigma_x$ \\
2D oscillating & $A = \begin{psmallmatrix} 0 & 1 \\ -0.9 & 1.8 \end{psmallmatrix}$ (eigs $0.9 \pm 0.3j$); $c = (0.5, 1)^\top$; $Q_w = \mathrm{diag}(0,1)$; $\Delta = 5.0$ \\
$A_{\text{fast}}$ & eigs $\{0.8, 0.5\}$; $c = (1, 0)^\top$; $Q_w = \mathrm{diag}(0,1)$; $\Delta = 0.75\sigma_x$ \\
$A_{\text{int}}$  & eigs $\{0.9, 0.8\}$; otherwise same as $A_{\text{fast}}$ \\
$A_{\text{slow}}$ & eigs $\{0.95, 0.9\}$; otherwise same as $A_{\text{fast}}$ \\
\midrule
\multicolumn{2}{@{}l}{\textbf{Sensing and channel}} \\
$R_v^{(1)}, R_v^{(2)}$ & $0.1$ (sensor 1); $0.4$ (sensor 2, two-sensor only) \\
Fading & Rayleigh ($\Omega_1 = 1.5$, sensor 1); Rician ($\Omega_2 = 0.7$, $K_2 = 3$; sensor 2) \\
Blockage duration & Weibull ($\kappa{=}2$; scale $5$/$3$ slots) \\
FBL (channel uses; info bits; Tx power / noise) & $N_B{=}128$; $L_B{=}256$; $p_{\text{tx}}/\sigma_n^2{=}17$\,dB \\
$\alpha_{\text{FP}}, \alpha_{\text{FN}}$ & $0.15$ (default); $\{0.05, 0.10\}$ also tested \\
\midrule
\textbf{Simulation Protocol} & $10$ seeds; $N_{\text{EP}}{\times}T_{\text{sim}} = 30{\times}300$ (episodes\,$\times$\,slots); CIs via $t$-dist.\ at $95\%$ \\
\bottomrule
\end{tabular}
}
\end{table}


\section{Empirical Evaluation}\label{sec:empirical}

This section validates Sec.~\ref{sec:main} empirically: the spectral characterization, the magnitude of the gain once unlocked, the effect of channel quality, and the two-sensor extension. Unless otherwise stated, parameters follow Table~\ref{tab:setup}, with results over ten seeds and $95\%$ CIs. We isolate the spectral effect with a single sensor; the two-sensor case follows in Sec.~\ref{sec:multi_sensor_extension}. Table~\ref{tab:dynamics_character} gives the per-system detection probabilities at $H{=}1$ and $H{=}5$ used throughout this section.

\subsection{Validation of the Spectral Characterization}

We first compare two systems on opposite sides of Theorem~\ref{thm:spectral}. The scalar system satisfies $\mathbf{c}^\top \mathbf{A} = \mathbf{A}\mathbf{c}^\top$ with $\lambda \geq 0$, and is therefore spectrally locked. In contrast, the oscillatory two-dimensional system has complex eigenvalues and admits no real left-eigenvector, making predictive lookahead structurally meaningful by Corollary~\ref{cor:complex}. For the scalar system (Table~\ref{tab:dynamics_character}), $P(L>0)=0.013\pm0.003$ is independent of $H$, confirming that predictive lookahead creates no new triggering opportunities. The oscillatory system, by contrast, climbs from $0.003$ at $H{=}0$ to $0.330$ at $H{=}7$ before saturating (Fig.~\ref{fig:dynamics_character}(a)). The gain accrues with diminishing returns and levels off by $H{\approx}5$, which is why we adopt $H{=}5$ as the working horizon; beyond it, additional lookahead yields little because the trigger lead rarely extends that far.

\subsection{Beyond Structural Unlock}

\begin{table}[t]
\centering
\caption{Dynamics-character sweep, single-sensor (clean channel, blockage onset $\PoutZero = 0.05$).}
\label{tab:dynamics_character}
\setlength{\tabcolsep}{3.5pt}
\resizebox{0.8\columnwidth}{!}{
\begin{tabular}{l c c c}
\toprule
System & $\rho(A)$ & $P(L{>}0; H{=}1)$ & $P(L{>}0; H{=}5)$ \\
\midrule
Scalar (0.9)        & $0.9$ (locked) & $0.013 \pm 0.003$ & $0.013 \pm 0.003$ \\
$A_{\text{fast}}$   & $0.8$ & $0.019 \pm 0.005$ & $0.019 \pm 0.005$ \\
$A_{\text{int}}$    & $0.9$ & $0.030 \pm 0.009$ & $0.032 \pm 0.011$ \\
$A_{\text{slow}}$   & $0.95$ & $0.104 \pm 0.025$ & $0.241 \pm 0.014$ \\
2D oscillating      & $0.95$ & $0.346 \pm 0.012$ & $0.524 \pm 0.022$ \\
\bottomrule
\end{tabular}
}
\vspace{-6pt}
\end{table}

The spectral condition determines whether lookahead helps, not how much. We compare three unlocked real-eigenvalue systems of increasing spectral radius with the oscillatory system. Fig.~\ref{fig:dynamics_character}(b) plots the $H{=}5$ gain against $\rho(\mathbf{A})^H$, the fraction of magnitude surviving $H$ steps (formalized in Remark~\ref{rem:operational}), with the underlying per-system values in Table~\ref{tab:dynamics_character}. Although all three real-eigenvalue systems are unlocked, their gains differ sharply: the fast-decaying system ($A_{\mathrm{fast}}$)  gains almost nothing, the intermediate ($A_{\mathrm{int}}$) little, and the slow-decaying system ($A_{\mathrm{slow}}$) substantially. The gain is negligible below $\rho^H\!\approx\!0.6$ and large above
$\approx\!0.7$. The locked scalar system sits at zero
gain despite a moderate $\rho^H$, confirming that magnitude alone is insufficient without
structural unlock. And at matched magnitude ($\rho^H\!=\!0.77$), the oscillating system exceeds the slow real one, a rotation bonus formalized next. Together, these observations show that passing the non-collinearity test is necessary but not sufficient for a large gain. Remark~\ref{rem:operational} formalizes the two requirements the figure exposes: sufficient prediction magnitude retained over the horizon, and sufficient separation between the current and predicted threshold directions, the latter supplied by rotation in the complex case.
\begin{figure}[!t]
    \centering
  \subfloat[]{%
       \includegraphics[width=0.7\linewidth]{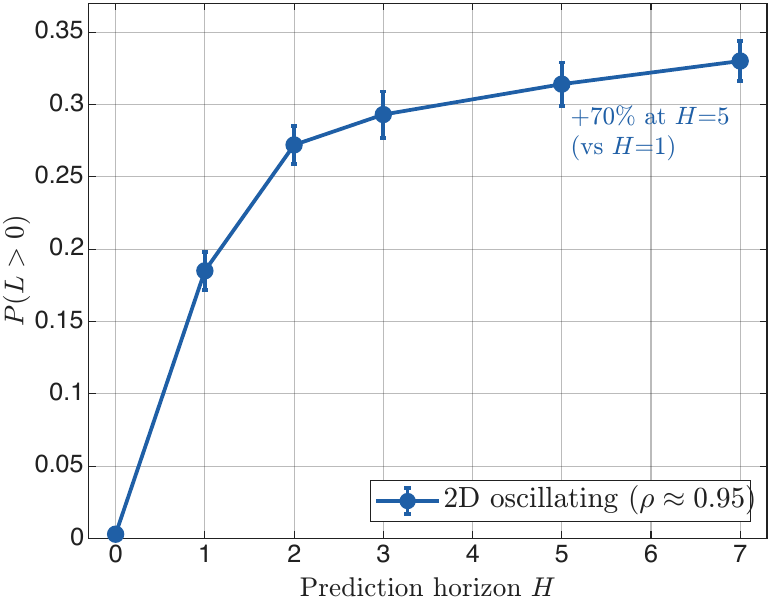}\label{fig:dynamics_character-a}}\\[-1pt]
  \subfloat[]{%
        \includegraphics[width=0.7\linewidth]{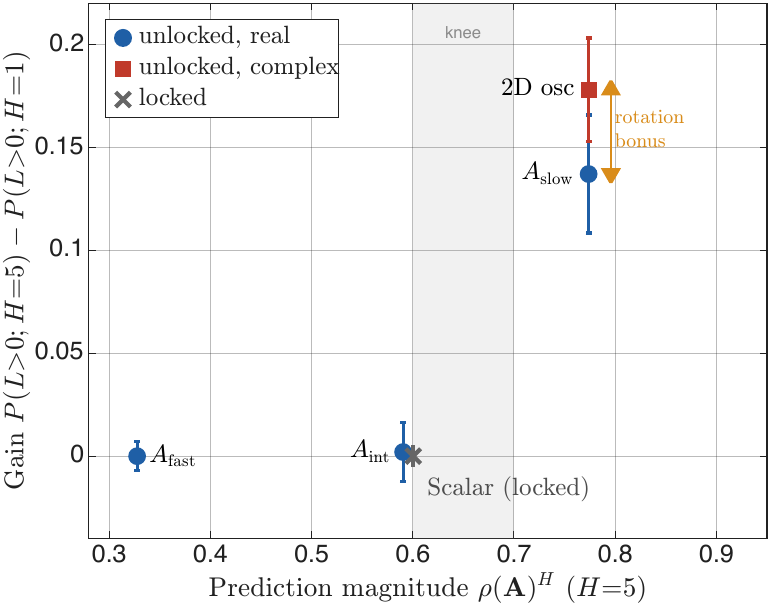}\label{fig:dynamics_character-b}}\\[-0.5pt]      
  \caption{Operational benefit beyond structural unlock. (a) 2D oscillating system, stressed channel ($\PoutZero{=}0.20$, $\rho \approx 0.95$): $P(L>0)$ grows with $H$. (b) Absolute gain $P(L{>}0;H{=}5) - P(L{>}0;H{=}1)$ versus retained prediction magnitude $\rho(\mathbf{A})^H$ at $H{=}5$: locked and low-magnitude systems cluster near zero, while gain emerges past the knee ($\rho(\mathbf{A})^H \approx 0.6$--$0.7$, shaded). At matched magnitude ($\rho(\mathbf{A})^H \approx 0.77$), the complex system (2D osc) exceeds the real one ($A_{\text{slow}}$) by the rotation bonus of Remark~\ref{rem:operational}.}
  \label{fig:dynamics_character}
  \vspace{-14pt}
\end{figure}  

\begin{remark}[Operational sufficiency, refined]\label{rem:operational}

The conditional variance of the predicted $z$-score given the current one is
\begin{equation}\label{eq:conditional_variance}
\mathrm{Var}(\hat z_j \mid \hat z_0)\! =\! \frac{1}{\sigma_j^2}\left[\mathbf{c}^\top \mathbf{A}^j \Pss (\mathbf{A}^j)^\top \mathbf{c} - \frac{(\mathbf{c}^\top \mathbf{A}^j \Pss \mathbf{c})^2}{\mathbf{c}^\top \Pss \mathbf{c}}\right].
\end{equation}
It vanishes when $(\mathbf{A}^j)^\top\mathbf{c}$ is collinear with $\mathbf{c}$ and is positive otherwise. For non-collinear systems, a large operational gain requires its standard deviation to be comparable to the confusion half-width $|\zminus|$, governed by a \emph{magnitude} factor $\rho(\mathbf{A})^H$ and a \emph{transversality} factor $|\sin\theta_H|$, where $\theta_H$ is the angle between $\mathbf{c}$ and $(\mathbf{A}^H)^\top\mathbf{c}$ in the $\Pss$ inner product. Small $\rho^H$ fails the first; near-alignment fails the second, while complex eigenvalues gain transversality through rotation.
\end{remark}

\subsection{Effect of Channel Quality}
We next evaluate Corollary~\ref{cor:channel}, which predicts that deeper lookahead becomes relatively more valuable as the channel quality deteriorates. Fig.~\ref{fig:channel_conditioning}(a) shows the results for the oscillatory system under increasing blockage onset $\PoutZero \in \{0.02, 0.10, 0.20, 0.30\}$ (clean, moderate, stressed, severe). As expected, the absolute value of $P(L>0)$ decreases as the channel worsens.  The relative benefit of deeper prediction, however, increases: the ratio $P(L>0;H{=}5)/P(L>0;H{=}1)$ grows from about $1.5$ on clean channels to $1.9$ under severe outage, matching the interpretation of Corollary~\ref{cor:channel}.

A control experiment isolates burst correlation. At the severe channel, replacing the Weibull-duration outage with an independent per-slot process at matched marginal probability raises the $H{=}5/H{=}1$ ratio from $1.87$ (bursty) to $1.99$ (independent). The small gap indicates that burst correlation accounts for only a minor part of the deeper-lookahead gain, so the independent-attempt idealization of Corollary~\ref{cor:channel} captures the effect well despite the correlated outage process.

The result is robust to both the decision tolerance and the threshold direction. Tightening the tolerance at the severe channel to $\alpha\in\{0.05,0.10\}$ widens the confusion region and lowers absolute $P(L>0)$, yet both benefits persist: the structural gain ($0.146$ at $H{=}5$ vs.\ $0.002$ at $H{=}0$, $\alpha{=}0.05$) and the deeper-lookahead ratio ($1.57$--$1.87$). Two further unlocked directions, $\mathbf{c}=(0.2,1)^\top$ and $(0.8,1)^\top$, give ratios of $1.88$, matching the baseline $1.87$.

\subsection{Two-Sensor Extension}\label{sec:multi_sensor_extension}

\begin{figure}[!t]
    \centering
  \subfloat[Single sensor]{%
       \includegraphics[width=0.49\linewidth]{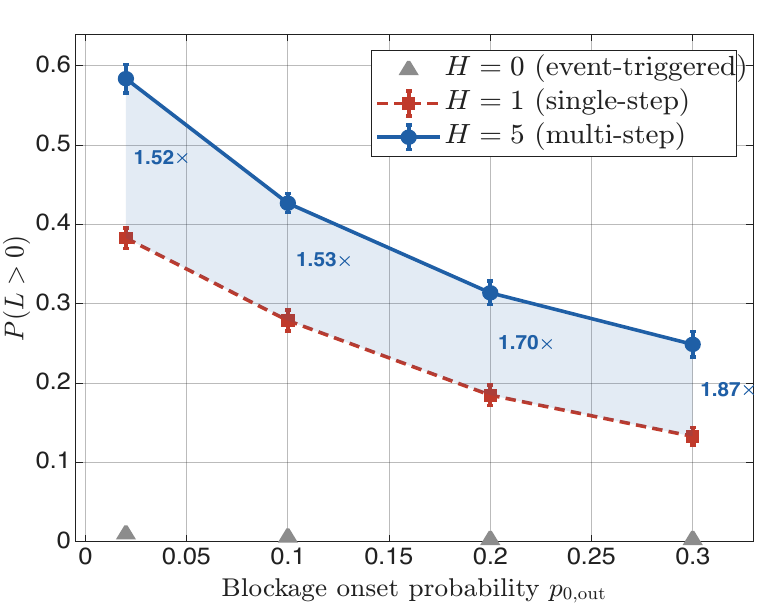}\label{fig:sota-a}}\hfill
  \subfloat[Two Sensors]{%
        \includegraphics[width=0.47\linewidth]{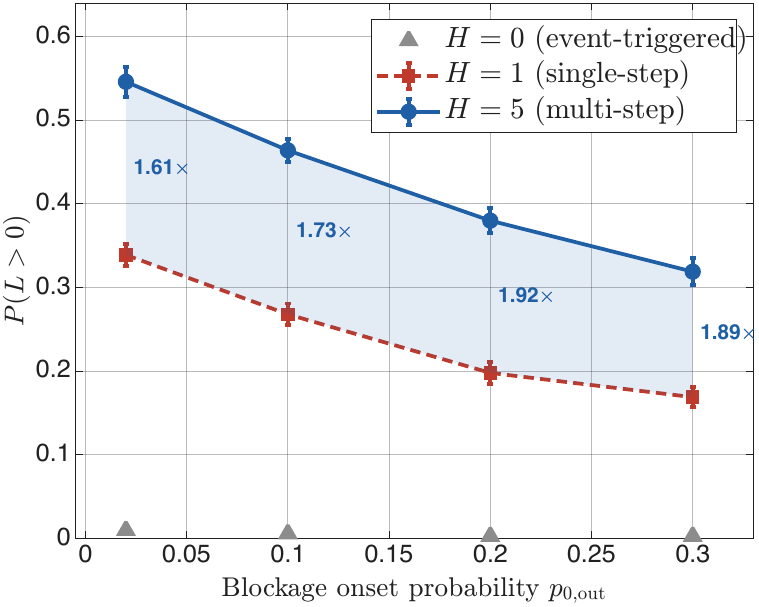}\label{fig:sota-b}}   
\caption{Early-detection probability $P(L{>}0)$ versus blockage-onset probability $\PoutZero$ for the 2D oscillating system, at horizons $H{=}0,1,5$. (a) single sensor; (b) two sensors. $H{=}0$ is the react-only (event-triggered) baseline; predictive lookahead ($H{\geq}1$) improves on it by an order of magnitude at low blockage ($\PoutZero{=}0.05$). Printed ratios give the deeper-lookahead gain $H{=}5/H{=}1$, which grows as the channel worsens.}
\vspace{-10pt}  \label{fig:channel_conditioning} 
\end{figure}  

Finally, we extend the framework to two sensors. Because only the marginal per-attempt failure probability enters Corollary~\ref{cor:channel}, the extension is
a direct test of the same theory: we add a noisier second sensor ($R_v^{(2)}=0.4$, Rician-faded) with its own independent link to the monitor, each running its own Kalman filter. The alarm decision uses the lower-variance posterior, while the transmitting sensor is chosen each slot with equal probability, so the delivered estimate is occasionally the noisier sensor's, trading some estimate quality for the independent outage path. Selecting the transmitting sensor based on the current link state could improve this further and is left to future work.

Fig.~\ref{fig:channel_conditioning}(b) shows the same qualitative trends as the single-sensor case, now structured by a sensing-channel tradeoff. On the clean channel, the single-sensor configuration performs slightly better, because mixing in the noisier second sensor degrades the remote estimate. As the channel worsens, the second sensor's independent outage path pulls two-sensor ahead: already at the moderate channel for $H{=}5$ ($0.463$ vs.\ $0.427$), and across both horizons on stressed and severe channels.

More importantly, the deeper-lookahead gain not only carries over but strengthens in
poor channels: mixing adds residual prediction uncertainty that deeper lookahead mitigates. This confirms that the channel-conditioned scaling of Corollary~\ref{cor:channel} is not limited to the single-sensor setting.

\subsection{Design Implications}

Combining the theoretical and empirical results yields a three-stage design rule:
\begin{enumerate}[leftmargin=*]
    \item \emph{Structural:} Apply the spectral test. If
          $\mathbf{c}^\top\mathbf{A}=\lambda\mathbf{c}^\top$ with $\lambda\geq0$,
          predictive lookahead is redundant, and the react-only rule ($H=0$) suffices.
    \item \emph{Magnitude:} For non-collinear unlocked systems, check $\rho(\mathbf{A})^H$ against the confusion-region width; across tested cases the transition lies near $\rho(\mathbf{A})^H \approx 0.6$--$0.7$ (Fig.~\ref{fig:dynamics_character}(b)), with complex eigenvalues amplifying the benefit through rotation.
    \item \emph{Operational:} Use larger horizons when communication reliability is
          poor, since the benefit of additional transmission opportunities grows with
          the packet-failure probability (Corollary~\ref{cor:channel}).
\end{enumerate}
In practice, oscillatory systems satisfy all three conditions and benefit most from predictive triggering.

\section{Conclusion}\label{sec:conclusion}
We characterized when predicting ahead is worthwhile before triggering an alarm over a lossy link. A one-line spectral test on $\mathbf{c}^\top\mathbf{A}$ and $\mathbf{c}^\top$ certifies when lookahead can create a new early-alarm opportunity; once it does, deeper prediction grows more valuable as the channel worsens. For non-collinear systems, structural unlock alone is not enough: prediction must also retain magnitude over the horizon, with oscillatory dynamics adding benefit through rotation. The resulting three-stage rule—check structure, magnitude, then channel—extends to two sensors through a sensing-channel tradeoff.

\bibliographystyle{IEEEtran}
\bibliography{references}

\end{document}